\documentclass[12pt]{article}
\usepackage{amsmath}
\usepackage{amssymb}
\usepackage{amsthm}
\usepackage{amsxtra}
\usepackage{enumerate}
\usepackage{hyperref}
\usepackage{amscd}
\usepackage{latexsym}
\usepackage{amsfonts}
\usepackage{graphicx}
\usepackage{lineno}
\usepackage{float}
\usepackage[all]{xy}
\usepackage{mathrsfs}

\usepackage{xcolor}
\usepackage{tikz}
\usetikzlibrary{tqft}

\newtheorem{theorem}{Theorem}[section]
\newtheorem{lemma}{Lemma}[section]
\newtheorem{corollary}{Corollary}[section]

\numberwithin{equation}{section}

\newcommand{\A}{\mathcal{A}}
\newcommand{\B}{\mathcal B}

\newcommand{\medn}{\medbreak \noindent}

\begin{document}

\title{Separability, Contextuality, and the Quantum Frame Problem
}

\author{Chris Fields\\
Allen Discovery Center\\
Tufts University\\
Medford, MA 02155, USA
\\
fieldsres@gmail.com\\
ORCID: 0000-0002-4812-0744 \\
\\and\\
\\James  F. Glazebrook
\\Department of Mathematics and Computer
Science \\
 Eastern Illinois University,
600  Lincoln Ave.\\ Charleston, IL 61920--3099, USA \\
jfglazebrook@eiu.edu
\\ Adjunct Faculty
\\ Department of Mathematics \\ University of Illinois at
Urbana--Champaign\\ Urbana, IL 61801, USA\\
ORCID: 0000-0001-8335-221X
}

\maketitle

\begin{abstract}
We study the relationship between assumptions of state separability and both preparation and measurement contextuality, and the relationship of both of these to the frame problem, the problem of predicting what does not change in consequence of an action.  We state a quantum analog of the latter and prove its undecidability.  We show how contextuality is generically induced in state preparation and measurement by basis choice, thermodynamic exchange, and the imposition of {\em a priori} causal models, and how fine-tuning assumptions appear ubiquitously in settings characterized as non-contextual.
\end{abstract}

\medn



\section{Introduction}\label{intro}

A theory exhibits contextuality, in the sense of Bell \cite{bell:66} or Kochen and Specker \cite{kochen:67}, if the outcome of some operation allowed by the theory can depend on what other operations are performed simultaneously \cite{Fine1982, peres:91, mermin:93, spekkens:05}.\footnote{This is sometimes called {\em quantum} or {\em intrinsic} contextuality to distinguish it from classical, causal, context dependence.}  Over the past decade, and in parallel with increasing recognition of its importance as a resource for quantum computing \cite{raussendorf:13, howard:14, vega:17, frembs:18, mansfield:18, amaral:19}, a number of formal representations of contextuality have been developed and their relevance to various operational settings investigated \cite{abramsky:11, abramsky:14b, abramsky:17, abramsky:18, cabello:14, cabello:15, cabello:16, cabello:21,dzha:17a, Dzha2017b, Dzha18, gudder:19, fg:21, fgm:22}; see the recent scholarly review by Adlam \cite{adlam:21} addressing several interpretative issues along with the various nuances raised by these and other related works.  Here we contribute to this progression of ideas by reformulating contextuality in terms of a decision problem that can be proved to be undecidable.  This decision problem is a quantum analog of the frame problem (FP) \cite{mccarthy:69, dietrich:96, nakashima:97, fields:13b, shanahan:16, dietrich:20}, a well-known problem in Artificial Intelligence (AI) that was first characterized, ironically, only a year after the publication of the Kochen-Specker theorem.  Informally, the FP is the problem of efficiently specifying what does not change in consequence of an action \cite{mccarthy:69}.  It is thus a problem about characterizing an action's side-effects, which can be expected to depend on the action's context.  Hence the FP can also be seen as the problem of efficiently specifying, or circumscribing, the context or, in AI terms, task environment that is in fact relevant to an action.

As in \cite{bell:66, kochen:67}, the operations or actions of interest here are measurement and, dually \cite{pegg:10}, preparation of a quantum state.  We represent these operations as physical interactions between a finite quantum system regarded as an ``agent'' ($A$ or Alice) and some other, disjoint finite quantum system ($B$ or Bob) that make use of some finite collection of finite quantum reference frames (QRFs  \cite{aharonov:84, bartlett:07}).  Following the development in \cite{fg:20a, addazi:21, fgm:21, ffgl:22, fgm:22b}, we assume that the joint state $|AB \rangle$ is separable over the time interval of interest, allowing us to write the interaction as:

\begin{equation} \label{ham}
H_{AB} = \beta_k k_B\, T_k \sum_i^N \alpha^k_i M^k_i,
\end{equation}
\noindent
where $k =~A$ or $B$, $k_B$ is Boltzmann's constant, $T_k$ is temperature, the $\alpha^k_i \in [0,1]$ are such that $\sum^N_i \alpha^k_i = 1$, the $M^k_i$ are $N$ Hermitian operators with eigenvalues in $\{ -1,1 \}$, and $\beta_k \geq \ln 2$ is an inverse measure of $k$'s thermodynamic efficiency that depends on the internal dynamics $H_k$.  We prove in \cite{fgm:22} that in any interaction given by Eq. \eqref{ham}, the QRFs deployed by either system $k$ can be represented, without loss of generality, by finite category-theoretic structures, cone-cocone diagrams (CCCDs) \cite{fg:19a} of Barwise-Seligman classifiers (or classifications) \cite{barwise:97}, that specify quantum computations implemented by $H_k$.  This representation effectively bundles the QRF employed to perform an operation (e.g. a length standard) together with the components of $H_{AB}$ that are employed (i.e. some subset of the $M^k_i$); we henceforth use the term `QRF' to refer to this combined representation.  We have previously shown, using this representation, that noncommutativity of QRFs induces contextuality \cite{fg:21}.  This result effectively reformulates previous criteria for representations of contextuality: noncommutativity of QRFs acting on sets of variables $\{ x_i \}$ and $\{ y_i \}$ corresponds both to the non-existence of a well-defined, well-behaved joint probability distribution over the $\{ x_i \}$ and $\{ y_i \}$ and to the non-existence of a well-defined section of a sheaf constructed jointly over the $\{ x_i \}$ and $\{ y_i \}$.

Here we first extend these results, in \S \ref{context-QRF}, by showing that contextuality is induced whenever Alice comprises mutually-separable components, i.e. components that maintain a separable joint state, that implement distinct QRFs \cite{fgm:23}.  We then state in \S \ref{QFP} a decision problem, the {\em Quantum Frame Problem} (QFP \cite{dietrich:20}), and prove that it is algorithmically undecidable, i.e. not decidable in finite time by a Turing machine with arbitrarily-large memory.  The QFP is, informally, the problem of determining whether an action alters the entanglement structure of the environment.  The undecidability of the QFP suggests that contextuality is ubiquitous, and that ``noncontextual'' settings depend for their definition on implicit fine-tuning assumptions.  We show, using a canonical Bell/EPR experiment as an example, that such hidden fine-tuning assumptions can mask contextuality.  We then review, in \S \ref{disc}, our previous results and those of others showing that contextuality can be induced by choice of basis and hence of QRF (\S \ref{choice}), by classical thermodynamic exchange (\S \ref{thermo}), and by the imposition of classical causal models (\S \ref{causal-models}).  We remark briefly in \S \ref{concl} on the relationship between contextuality and G\"{o}del's \cite{godel:31} celebrated 1st incompleteness theorem, and conclude that contextuality is in no way paradoxical, but is rather to be expected whenever finite observers measure (or prepare) multiple distinct degrees of freedom in any sufficiently-large environment.
\footnote{Contextuality in quantum mechanics is by now accepted as a state of affairs to be lived with. Curiously, noncontextuality seems to be an equally perplexing issue. Hofer-Szab\'{o} in \cite{hofer:22} distinguishes between {\em simultaneous} and {\em measurement noncontextuality} as distinct, logically
independent models. In the first case, an ontological (hidden variable) model is noncontextual if every ontic (hidden) state determines the probability of the outcomes of every measurement independently of whatever other measurements are simultaneously performed (otherwise the model is contextual). In the second case, an ontological model is noncontextual if any two measurements represented by the same self-adjoint operator, or equivalently, which have the same probability distribution of outcomes in every quantum state, also have the same probability distribution of outcomes in every ontic state (ditto). This distinction reflects upon the schism between operational and ontological models when dealing with contextuality, as exemplified by the fact that contextuality-by-default \cite{dzha:17a,Dzha2017b} and causal contextuality \cite{cavalcanti:18,gogioso:23a,gogioso:23b}, respectively, are different theories \cite{hofer:21}.}

\section{Separable QRFs induce contextuality} \label{context-QRF}

We begin by supposing that Alice comprises (at least) two distinct components $A_1$ and $A_2$ that implement distinct QRFs $Q_1$ and $Q_2$, respectively.  The QRFs $Q_1$ and $Q_2$ can be regarded as (quantum) computations, physically implemented by the internal interactions $H_{A_1}$ and $H_{A_2}$, respectively, that measure and dually, prepare states $|S_1 \rangle$ and $|S_2 \rangle$ (over time, densities $\rho_{S_1}$ and $\rho_{S_2}$) of distinct sectors $S_1$ and $S_2$, respectively, of the boundary $\mathscr{B}$ separating $A$ from its environment $B$; see \cite{fgm:22, fgm:21, fgm:22b} for details.  Then we can state:

\begin{lemma} \label{sep-lemma}
If systems $A_1$ and $A_2$ implement QRFs $Q_1$ and $Q_2$, respectively, and the joint state $|A_1 A_2 \rangle$ $($or density $\rho_{A_1 A_2})$ is separable as $|A_1 A_2 \rangle = |A_1 \rangle |A_2 \rangle$ $(\rho_{A_1 A_2} = \rho_{A_1} \rho_{A_2})$, then the commutator $[Q_1,Q_2] \neq 0$.
\end{lemma}

\begin{proof}
We recall from \cite{fgm:22}, Thm. 1 that every QRF can be specified by a CCCD, i.e. a commutative diagram of the form:

\begin{equation}\label{cccd-vae}
\begin{gathered}
\xymatrix@C=6pc{\mathcal{A}_1 \ar[r]_{g_{12}}^{g_{21}} & \ar[l] \mathcal{A}_2 \ar[r]_{g_{23}}^{g_{32}} & \ar[l] \ldots ~\mathcal{A}_k \\
&\mathbf{C^\prime} \ar[ul]^{h_1} \ar[u]^{h_2} \ar[ur]_{h_k}& \\
\mathcal{A}_1 \ar[ur]^{f_1} \ar[r]_{g_{12}}^{g_{21}} & \ar[l] \mathcal{A}_2 \ar[u]_{f_2} \ar[r]_{g_{23}}^{g_{32}} & \ar[l] \ldots ~\mathcal{A}_k \ar[ul]_{f_k}
}
\end{gathered}
\end{equation}
\noindent
in which the $\mathcal{A}_i$ are Barwise-Seligman classifiers, the maps $f_i$, $h_i$, and $g_{ij}$ are Barwise-Seligman infomorphisms (i.e. morphisms of classifiers), and the ``core'' $\mathbf{C^\prime}$ is a Barwise-Seligman classifier that is simultaneously the category-theoretic colimit of the ``incoming'' maps $f_i$ and the category-theoretic limit of the ``outgoing'' maps $h_i$.  Details of this construction following the aforesaid concepts as established in \cite{barwise:97}, are provided in \cite{fg:19a}, and reviewed with further developments in \cite{fgm:22, fgm:21, fgm:22b}.  If the QRFs $Q_1$ and $Q_2$ commute, their corresponding CCCDs commute; in this case, the cores $\mathbf{C^\prime_1}$ and $\mathbf{C^\prime_2}$ can be mapped to a common core $\mathbf{C^\prime_{1,2}}$ that is the colimit of all incoming maps, and limit of all outgoing maps, thus defining a single CCCD over the combined sets $\{ \mathcal{A}_{1,i} \}$ and $\{ \mathcal{A}_{2,i} \}$ of classifiers of $Q_1$ and $Q_2$.  This CCCD specifies a single QRF $Q_1 Q_2 = Q_2 Q_1$.

We now recall from \cite{fgm:22}, Thm. 2 that the action of any QRF $Q$ induces a topological quantum field theory (TQFT) $\mathscr{T}_Q$ representable as a cobordism of linear maps with copies of the Hilbert space $\mathcal{H}_S$ of the sector $S$ on which $Q$ acts as boundaries.  As the data generated by $Q$ at discrete times $t_i, t_j$ are just $|S(t_i) \rangle$ and $|S(t_j) \rangle$, respectively, and as the TQFT $\mathscr{T}_Q |_{t_i \rightarrow t_j} : |S(t_i) \rangle \mapsto |S(t_j) \rangle$ by definition, we can simply identify $Q$ with $\mathscr{T}_Q$ \cite{fgm:23}.

Provided $Q_1$ and $Q_2$ commute, therefore, there is a TQFT $\mathscr{T}_{Q_1 Q_2} = \mathscr{T}_{Q_2 Q_1}$ that can be identified with the combined QRF $Q_1 Q_2 = Q_2 Q_1$; this TQFT can be represented by:

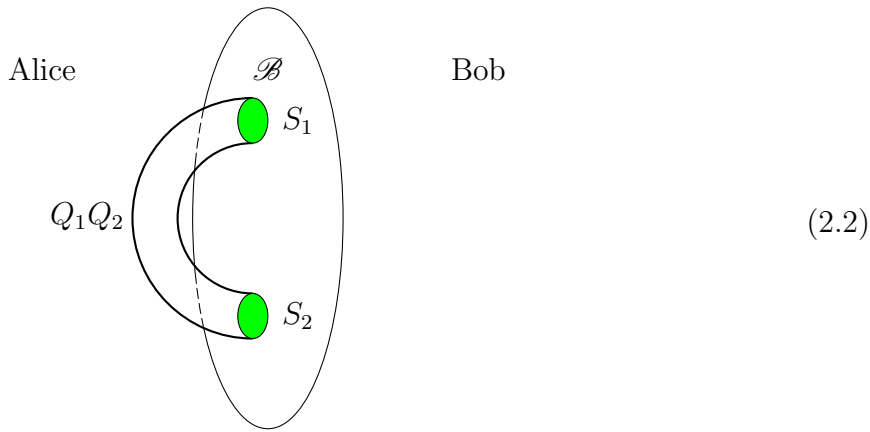
\begin{equation} \label{bipartite-ops}
\begin{gathered}
\begin{tikzpicture}[every tqft/.append style={transform shape}]
\draw[rotate=90] (0,0) ellipse (2.8cm and 1 cm);
\node[above] at (0,1.7) {$\mathscr{B}$};
\draw [thick] (-0.2,1.6) arc [radius=1.6, start angle=90, end angle= 270];
\draw [thick] (-0.2,1) arc [radius=1, start angle=90, end angle= 270];
\draw[rotate=90,fill=green,fill opacity=1] (1.3,0.2) ellipse (0.3 cm and 0.2 cm);
\draw[rotate=90,fill=green,fill opacity=1] (-1.3,0.2) ellipse (0.3 cm and 0.2 cm);
\node[above] at (-3,1.7) {Alice};
\node[above] at (2.8,1.7) {Bob};
\node[above] at (-2.4,-0.3) {$Q_1 Q_2$};
\node at (0.4,1.3) {$S_1$};
\node at (0.4,-1.3) {$S_2$};
\draw [ultra thick, white] (-0.8,1.3) -- (-1,1.3);
\draw [ultra thick, white] (-0.8,1.1) -- (-1,1.1);
\draw [ultra thick, white] (-0.8,0.9) -- (-1,0.9);
\draw [ultra thick, white] (-0.8,-0.9) -- (-1,-0.9);
\draw [ultra thick, white] (-0.8,-1.1) -- (-1,-1.1);
\draw [ultra thick, white] (-0.8,-1.3) -- (-1,-1.3);
\end{tikzpicture}
\end{gathered}
\end{equation}

where Alice and Bob are represented as the components of the bulk to the left and right of $\mathscr{B}$, respectively.  In this case Alice implements, and the components $A_1$ and $A_2$ jointly implement, a single unitary process $Q_1 Q_2$.  The joint state $|A_1 A_2 \rangle$ of the components $A_1$ and $A_2$ is, therefore, entangled over the entire time period of interest.

Separability of $|A_1 A_2 \rangle$ can, therefore, only be achieved if $Q_1$ and $Q_2$ do not commute, i.e. $[Q_1,Q_2] \neq 0$.
\end{proof}

We note that the converse of Lemma \ref{sep-lemma} is true trivially: if $Q_1$ and $Q_2$ do not commute, the joint state $|A_1 A_2 \rangle$ cannot be entangled.  We show in \cite{fgm:23} that the value of the commutator $[Q_1,Q_2]$ can be related to the entanglement of formation \cite{wootters:01} of mixed states of $B$ prepared by the joint action of $Q_1$ and $Q_2$.

If $Q_1$ and $Q_2$ commute, and hence the joint operation $Q_1 Q_2$ is well-defined as a QRF, then $Q_1$ and $Q_2$ are, in the language of \cite{fg:21}, {\em co-deployable} and the joint probability distribution over the states $|S_1 \rangle$ and $|S_2 \rangle$ of their respective sectors $S_1$ and $S_2$ is well-defined.  If $Q_1$ and $Q_2$ do not commute, and hence the joint operation $Q_1 Q_2$ is not well-defined as a QRF, then $Q_1$ and $Q_2$ are {\em non-codeployable} and no joint probability distribution over the states $|S_1 \rangle$ and $|S_2 \rangle$ of their respective sectors $S_1$ and $S_2$ is well-defined.  This latter condition corresponds to contextuality, as shown explicitly in \cite{fg:21} using methods developed in \cite{Dzha2017b, Dzha18}.  Hence we have:

\begin{theorem} \label{sep-th}
If systems $A_1$ and $A_2$ implement QRFs $Q_1$ and $Q_2$, respectively, and the joint state $|A_1 A_2 \rangle$ $($or density $\rho_{A_1 A_2})$ is separable as $|A_1 A_2 \rangle = |A_1 \rangle |A_2 \rangle$ $(\rho_{A_1 A_2} = \rho_{A_1} \rho_{A_2})$, then observational outcomes obtained with, and state preparations implemented with, the QRFs $Q_1$ and $Q_2$ will exhibit contextuality.
\end{theorem}

\begin{proof}
We recall from \cite{fg:21}, Thm. 7.1 that non-commutativity, and hence non-codeployability of QRFs induces contextuality.  Whenever the conditions of Lemma \ref{sep-lemma} are met, the relevant QRFs are non-codeployable.
\end{proof}

As shown in \cite{fgm:22}, the set of all finite CCCDs that specify QRFs form a category $\mathbf{CCCD}$; the morphisms of this category compose CCCDs specifying QRFs to form larger CCCDs, and decompose CCCDs into components that are themselves CCCDs.  We can therefore consider the notion of a quiver representation for CCCDs, referring to \cite{seigal:22} for relevant definitions.  We have from \cite{fgm:22}, Thm. 9:
\begin{theorem}\label{section-th}
If a diagram of the form of Diagram \eqref{cccd-vae}, when viewed as a quiver, is noncommutative (and hence is not a CCCD), then at least one section of its quiver representation by vector spaces is not definable; in particular, at least one section of its representation by vector spaces of measurable functions is not definable.
\end{theorem}
We note that Theorems \ref{sep-th} and \ref{section-th} describe conditions that, in the spirit of \cite{kochen:67}, induce break-downs of logical consistency; in particular, break-downs of logical consistency between the classifiers composing the CCCD representations of the relevant QRFs. They also bear obvious similarity to the main result of \cite{abramsky:11} proving that in some `measurement scenario', the non-existence of a global section of a sheaf of measurable distributions implies Kochen-Specker contextuality.\footnote{It is worth pointing out that a distributed system of information flow of \cite{barwise:97}, of which CCCDs (when considered as systems of logical gates) are an example, already embody notions of context and causation, and are especially suited for modeling ontologies. These characteristics have been professed in \cite{collier:11,seligman:09}, in which the former uses the logical formulism of \cite{barwise:97} to argue that causation itself may be viewed as a form of computation resulting from the regular relations in a distributed system (such as a CCCD), and the latter shows that for any pair of classifiers $\A,\B$ occurring in such a system, there exists some (logic) infomorphism between them such that $\A$ directly {\em causes} $\B$. In a companion theory of integrative information involving abstract logical structures known as {\em Institutions} (see e.g. \cite{goguen:05a,goguen:05c,kent:18}), the element of context is further emphasized. As briefly recalled in the proof of Lemma \ref{sep-lemma}, the basic mechanism of a QRF in conjunction with binary-valued classifiers provides a model for a generic topological quantum field theory \cite{fgm:22} with no need for further logical abstractions. The `Institutional' approach will, however, be addressed relative to the present development of ideas at a later date.}

From a physical perspective, the mechanism enforcing contextuality is clear from Diagram \eqref{bipartite-ops}: the QRFs $Q_1$ or $Q_2$ act, via their respective sectors $S_1$ and $S_2$ of $\mathscr{B}$, on the quantum system $B$.  Any such actions affect the state $|B \rangle$, of which the sector states $|S_1 \rangle$ and $|S_2 \rangle$ are effectively (via the actions of QRFs implemented by $B$) samples.  Hence if $Q_1$ and $Q_2$ do not commute, $A$'s actions with either render the outcomes of measurements made with the other unpredictable.  The unpredictability of side-effects of actions is the core of the FP \cite{mccarthy:69}.  Indeed, we have previously proved that a distributed information-flow system (e.g. a CCCD) can be informationally unencapsulated (i.e. embody an ``open environment'' solution to the FP) only in the absence of intrinsic contextuality  \cite[Cor 8.1]{fg:21}.  We can, therefore, expect that contextuality generates a quantum analog of this well-known, provably undecidable \cite{dietrich:20} problem.

\section{Contextuality renders the QFP undecidable} \label{QFP}

\subsection{Statement of the QFP} \label{QFP-statement}

The FP was originally formulated as the problem of efficiently characterizing the components of the overall state of the environment (in AI language, the states of and relations between objects in the environment) that do not need to be updated in consequence of an action \cite{mccarthy:69}.  An efficient means of predicting all side-effects of any action taken in any environment would solve the FP.  It is generally recognized that all practical solutions of the FP are heuristic, and rely on local, classical, causal models of the task environment \cite{nakashima:97, fields:13b}.  Local causal models being insufficient to predict the side-effects of an action naturally suggests the influence of hidden variables.  In a quantum setting, such hidden variables may be nonlocal in the sense of acting via (unobserved components of) the entanglement structure of the environment.  The simplest measure of entanglement structure of an environment $B$ is the {\em entanglement entropy}:
\begin{equation}\label{entanglement-1}
\mathcal{S}(B) =_{def} \mathrm{max}_{B_1, B_2 | B_1 B_2 = B} ~\mathcal{S} (|B_1 B_2 \rangle)
\end{equation}
Hence we can state, for any agent $A$ implementing an action $a$ on its environment $B$:

\begin{center}
{\bf QFP:} Does an action $a$ on an environment $B$ change the entanglement entropy $\mathcal{S}(B)$?
\end{center}

If the QFP can be answered with `no', local causal models may be sufficient, in practice, to solve the FP; hence the latter can be thought of, at least in practice, as the problem of efficiently modeling local causal relations.  If the QFP is answered `yes', all such local causal models are insufficient to solve the FP.

\subsection{The QFP is undecidable} \label{QFP-undec}

The FP has been shown \cite{dietrich:20} to be equivalent to the Halting Problem (HP), the problem of deciding whether an arbitrary algorithm halts on an arbitrary input \cite{turing:37, hopcroft:79}.  As the HP is known to be algorithmically undecidable, the interesting direction of proof is HP $\rightarrow$ FP.  This is straightforward: if whether an arbitrary algorithm halts on an arbitrary input is undecidable, whether the FP solution of an arbitrary agent $A$ halts given an arbitrary action $a$ is undecidable.

That the QFP is algorithmically undecidable is also straightforward \cite{dietrich:20}.  Decidability in finite time is equivalent to decidability given a finite number of finite inputs, i.e. finite-resolution observational outcomes.  The idea of ``measurement'' as a process that yields observational outcomes is physically meaningful only if the interacting systems $A$ and $B$ are separable, i.e. only if $\mathcal{S} (|A B \rangle) = 0$; otherwise Eq. \eqref{ham} fails to describe the $A$-$B$ interaction, and hence $A$ cannot be regarded as deploying a well-defined QRF or as recording an observational outcome as classical data.  Separability can, however, only be achieved if the $A$-$B$ interaction $H_{AB}$ is weak, which in turn requires that dim($\mathcal{H}_A$), dim($\mathcal{H}_B) \gg$ dim($H_{AB}$).  This is effectively a ``large environment'' assumption that allows the existence of nonlocal hidden variables in $B$.  Finite classical data obtainable by $A$ cannot, in this case, determine $\mathcal{S}(B)$ or even dim($B$); hence it cannot determine whether $\mathcal{S}(B)$ is altered by an action $a$ \cite{dietrich:20, fgm:21}.  Here we provide an alternative proof of this result that does not rely explicitly on a large environment, but rather highlights the role of contextuality.

\begin{theorem} \label{QFP-th}
The QFP is algorithmically undecidable.
\end{theorem}

\begin{proof}
The entanglement entropy $\mathcal{S}(B)$ of Eq. \eqref{entanglement-1}clearly cannot be measured using just one QRF; hence given the results of \S \ref{context-QRF} above, $\mathcal{S}(B)$ cannot be measured using any combination of mutually-commuting QRFs.  Hence the only case of interest is that in which $A$ comprises separable components $A_i$ that, as required by Lemma \ref{sep-lemma}, implement QRFs (or sets of mutually-commuting QRFs) that do not mutually commute.  In general, any QRF $Q_i$ deployed by a component $A_i$ acts on some component of $A_i$'s environment, which is the joint system $B \prod_j A_j$ where $j \neq i$.  For simplicity, we consider only two components $A_1$ and $A_2$ that are {\em minimal} in that each only a single QRF (or set of mutually-commuting QRFs), $Q_1$ and $Q_2$ respectively, where $Q_1$ and $Q_2$ do not commute.  We also assume $H_{A_1 A_2} = 0$, so each of these QRFs acts only on $B$.  We can, in this case, consider the situation from the perspective of either of the minimal, separable components $A_1$ or $A_2$.  As each of these components has only a single (effective) QRF, neither component alone can measure $\mathcal{S}(B)$ or decide the QFP.  This same argument applies, however, also to any further minimal, separable components of $A$ that interact with either $B$ or any of the other $A_i$ (equivalently, to any further modular quantum computations implemented by $A$ that accept classical inputs from either $B$ or the $A_i$).  The entanglement entropy $\mathcal{S}(B)$ cannot, therefore, be measured by $A$, so the QFP cannot be decided by $A$.  As $A$ is an arbitrary finite system, the QFP is algorithmically undecidable.
\end{proof}

From a physical perspective, Lemma \ref{sep-lemma} restricts any minimal, separable component of a system $A$ to the use of its own computational resources -- which must all mutually commute -- to measure the state of its own environment, which may comprise other components of $A$.  Effectively restricting each separable component to a single QRF renders the environment of each such component ``large enough'' to prevent measurement of its entanglement entropy.

Theorem \ref{QFP-th} has two important corollaries:

\begin{corollary} \label{self-cor}
No finite system $A$ can measure the entanglement entropy $\mathcal{S}(|AB \rangle)$ across the boundary $\mathscr{B}$ that separates it from its environment $B$.
\end{corollary}

\begin{proof}
Any finite system $A$ comprises some set of minimal, separable components, so it is enough to assume that $A$ itself is a minimal, separable system.  Lemma \ref{sep-lemma} restricts any minimal, separable system $A$ to a single QRF acting the boundary $\mathscr{B}$ that separates it from its environment $B$.  The maximum information obtainable is the state $|S \rangle$ of the sector $S \subseteq \mathscr{B}$ acted upon by this QRF \cite{fgm:22, fgm:21, fgm:22b}.
\end{proof}

Corollary \ref{self-cor} in fact follows directly from Eq. \eqref{ham}; the maximum information obtainable by $A$ at time $t$ is the $N$-bit encoding of an eigenvalue of $H_{AB}$.  We state it as a corollary of Theorem \ref{QFP-th} to emphasize its relation to the QFP: no finite system can determine that it is separable from, and hence only in classical causal interaction with, its environment.  The eigenvalues of $H_{AB}$ are also, clearly, insufficient to specify the entanglement entropy $\mathcal{S}(A)$; hence $A$ cannot determine whether it has separable components.

\begin{corollary} \label{isolation-cor}
No finite system $A$ can determine that any decomposition of its environment $B = B_1 B_2$ isolates all causal consequences of its actions in a single component $B_1$.
\end{corollary}

\begin{proof}
Isolating all causal consequences of its actions in $B_1$ requires determining that $\mathcal{S} (|B_1 B_2 \rangle) = 0$, which Theorem \ref{QFP-th} forbids.
\end{proof}

Corollary \ref{isolation-cor} shows that no system can isolate its own task environment.  The undecidability of the FP immediately follows.  More significantly for our purposes, no system $A$ can determine the context of any of its actions $a$.  Griffiths \cite{griffiths:19}, for example, shows that preparations and measurements become noncontextual when described in a consistent, i.e. commuting framework.  A framework is a context, extended over time.  Corollary \ref{isolation-cor} shows that no finite system can specify a consistent framework for its own actions.

\subsection{Example: Bell/EPR experiments} \label{bell-expt}

The most obvious example of measuring an entanglement entropy is a Bell/EPR experiment.  Such experiments have obviously been performed, and have obviously yielded results, some of the most important in contemporary physics \cite{aspect:81, aspect:82, aspect:99, bartosik:09, hensen:15, georgescu:21}.  Theorem \ref{QFP-th} implies, however, that such experiments can only be regarded as measurements of entanglement entropies subject to further assumptions that place {\em a priori} limits on contextuality.  It is instructive to examine these assumptions explicitly.

We can represent the set-up of a Bell/EPR experiment by the following state-space decomposition:

\begin{equation} \label{bell-setup}
\begin{gathered}
\begin{tikzpicture}
\draw [thick] (0,0) -- (0,4) -- (3,4) -- (3,0) -- (0,0);
\draw [thick] (1,1) -- (1,3) -- (3,3) -- (3,1) -- (1,1);
\draw [thick] (2,1) -- (2,3);
\draw [thick] (1,2) -- (2,2);
\node at (0.5,2) {$C$};
\node at (2.5,2) {$B$};
\node at (1.5,1.5) {$A_2$};
\node at (1.5,2.5) {$A_1$};
\end{tikzpicture}
\end{gathered}
\end{equation}

We interpret $A_1$ and $A_2$ as the observers, $B$ as the entangled pair together with the detectors, and $C$ as everything else, including whatever prepares the entangled pair, any 3rd parties, and the general environment.  We assume $H_{BC} \sim 0$; in particular, that it is insufficient to induce decoherence.  The pointer states of the detectors are, in this case, decohered only by interaction with the observers $A_1$ and $A_2$, consistent with the model of QRF-driven decoherence in \cite{fgm:21, fgm:22b}.

Each of $A_1$ and $A_2$ interacts with only one detector, so as is well-known, neither can observe entanglement in isolation.  The joint systems $A = A_1 A_2$ and $A C$ observe both detectors, so it is these systems that, in standard interpretations of the Bell/EPR set-up, observe entanglement; ``Charlie'' ($C$) becomes the locus of entanglement observation if $A_1$ and $A_2$ separately report their observations.  Hence the interactions $H_{A_1 A_2}$, $H_{A_1 C}$, and $H_{A_2 C}$ are critical for observing entanglement.

Assuming that $A_1$ and $A_2$ are separable, their interaction $H_{A_1 A_2}$ can be represented as classical communication \cite{fgm:22b}; indeed they are standardly represented as exchanging classically-encoded data once their observations are completed.  We can assume that this communication is free of classical noise.  Assuming that it is free of {\em quantum} noise, however, requires assuming that $A_1$ and $A_2$ share a reference frame, e.g. a $z$ axis if $s_z$ measurements are employed to communicate outcomes encoded as bit strings.  This reference frame must be physically implemented, so is a QRF.  Jointly implementing a QRF, however, induces entanglement as illustrated in Diagram \eqref{bipartite-ops}; see \cite{fgm:21, fgm:22b, fgm:23} for proofs and further discussion.  Equivalently, separability of $A_1$ and $A_2$ requires each to have free choice of basis for every measurement, including free choice of $z$ axis (i.e. ``language'') when communicating their results \cite{fgm:22b}.  Hence the common assumption that post-observation classical communication is conceptually unproblematic involves an assumption of superdeterminism: $A_1$ and $A_2$ must be assumed to choose the same basis when they prepare and measure classical encodings.  That ``classical communication'' always involves quantum measurement in this way has been noted previously \cite{tipler:14}; one can make similar remarks, clearly, about requirements for quantum measurements and prior agreement about choice of basis when multiple observers calibrate their detectors using a common standard.

While superdeterminism of choice of basis for post-observation classical communication is required to measure $\mathcal{S}(B)$, superdeterminism of choice of basis for the observations themselves (i.e. choice of ``settings'') prevents measurement of $\mathcal{S}(B)$ by rendering $A_1$ and $A_2$ entangled \cite{mermin:93}.  We can, therefore, see superdeterminism of choice of basis for post-observation classical communication as a specific, {\em a priori, ceteris paribus} or ``nothing else changes'' \cite{nakashima:97} assumption that is required to render $A_1$ and $A_2$ able to perform a joint measurement \cite{bohr:28, mermin:18}.  Assumptions that circumscribe the task environment by ruling out external influences are the basis of all heuristic solutions to the FP \cite{dietrich:96, nakashima:97}.\footnote{The approach to causality (and hence to the FP) in \cite{nakashima:97} is via the theory of situations (i.e. contexts), the latter being cast in \cite{seligman:09} within the category of information flow of classifications \cite{barwise:97}.}

Such {\em ceteris paribus} assumptions are clearly assumptions of {\em noncontextuality}: the larger context in which the experiment is performed is assumed to be constrained in a way that specifically superdetermines a particular classical, causal process -- classical information exchange between $A_1$ and $A_2$ -- while not superdetermining anything else.  This {\em a priori} constraint is a clear instance of fine-tuning.  Indeed Cavalcanti \cite{cavalcanti:18} has shown that any classical causal model that violates a Bell inequality or demonstrates Kochen-Specker contextuality will involve fine-tuning, and hence will violate an essential principle of the classical causal modeling framework, namely, that such models should not be fine-tuned. A counterpoise, as pointed out in  \cite{shrapnel:19}, is the result of Wood and Spekkens \cite{wood:15} asserting that any classical causal model reproducing the observed statistics of the Bell correlations must be fine-tuned (equivalently `unfaithful'); see \cite{adlam:21} for review and further discussion.

\section{Induced contextuality} \label{disc}

The results of \S \ref{context-QRF} and \ref{QFP} above suggest that {\em all} multi-QRF measurements, and in particular all measurements made by observers comprising multiple, separable components, are contextual.  In this case, contextuality becomes not just a default assumption as in \cite{Dzha2017b}, but rather a principled requirement.  If this is the case, all measurement settings considered ``non-contextual'' involve implicit assumptions, e.g. assumptions of fine-tuning.

As discussed in \cite{fg:21}, the ubiquity of contextuality is suggested by the very existence of the FP: if task environments could be circumscribed non-heuristically, or even reliably circumscribed heuristically, the FP would not arise.  Practical ``solutions'' of the FP are, however, based on heuristics of unknown reliability.  This is because task environments cannot, in practice, be fully isolated or fully characterized in advance.  Unexpected side-effects cannot, therefore, be ruled out.

While the FP is often attributed to the environment being ``large'' or ``open,'' we have seen in the above discussion that it in fact follows from the physics of measurement, and thus characterizes all environments.  Here we consider three specific contextuality-inducing mechanisms in more detail.

\subsection{Basis choice induces contextuality} \label{choice}

It is well-known that the observability of entanglement depends on choice of basis \cite{zanardi:01, zanardi:04, torre:10, harshman:11, thirring:11}; a Bell state is separable if measured in the Bell basis $(|10 \rangle \pm |01 \rangle)/\surd 2$.  If $A_1$ and $A_2$ deploy a Bell basis, however, they are {\em ipso facto} entangled; choosing the Bell basis merely transfers the entanglement from the measured state to the observers.  Theorem \ref{sep-th} generalizes this to arbitrary QRFs: if $A_1$ and $A_2$ deploy non-commuting QRFs (and hence are separable), they will observe contextuality; if they deploy commuting QRFs (and hence are entangled), they will not.

The classical limit of non-contextual observations by separable observers -- e.g. spatially-separated observers -- generalizes the assumption of fine-tuning of classical communication found in the case of Bell/EPR experiments.  Observations made by distinct, spatially-separated observers are guaranteed to commute only if 1) the observers deploy the same QRFs, and 2) the environment being measured is causally uniform, i.e. homogeneous and isotropic.  The existence of equivalent observers and the causal uniformity of the environment are central, if largely implicit, postulates of the classical worldview.  As they require precise values of relevant constants and rapid dissipation of early inhomogeneities, they are effectively fine-tuning assumptions, as is broadly acknowledged \cite{barnes:22, frankel:22}.

\subsection{Thermodynamic exchange induces contextuality} \label{thermo}

Irreversibly recording classical information requires free energy \cite{parrondo:15}: hence all information processing that results in persistently-recorded data requires thermodynamic exchange with a free-energy source that we can assume, without loss of generality, to be external to any system $A$ that acts as an observer.  This can be made explicit with the state-space decomposition:

\begin{equation} \label{thermo-ex}
\begin{gathered}
\begin{tikzpicture}
\draw [thick] (0,0) -- (0,3) -- (3,3) -- (3,0) -- (0,0);
\draw [thick] (1,3) -- (1,0);
\draw [thick] (2,3) -- (2,0);
\draw [thick] (1,1.5) -- (2,1.5);
\node at (0.5,1.5) {$A_{mp}$};
\node at (2.5,1.5) {$B$};
\node at (1.5,0.8) {$A_i$};
\node at (1.5,2.2) {$A_{th}$};
\end{tikzpicture}
\end{gathered}
\end{equation}

in which the component $A_{th}$ engages in thermodynamic exchange with the environment $B$, the $A_i$ engage in preparation and measurement of $B$ as in Diagram \eqref{bell-setup} above, and $A_{mp}$ is a metaprocessing component that performs tasks such as outcome recording and control-flow management.  The internal interactions $H_{A_{th} A_i}$ and $H_{A_{th} A_{mp}}$ implement internal thermodynamic exchange.  The overall $A$-$B$ interaction is $H_{AB} = H_{A_{th} B} + H_{A_{i} B}$.

The measuring components $A_i$ in Diagram \eqref{thermo-ex} have no access to the thermodynamic interaction $H_{A_{th} B}$ and cannot determine its effects on $B$ \cite{fg:20a, fgm:21}.  Hence $H_{A_{th} B}$ defines, in effect, an unknown context for the measurement interaction $H_{A_{i} B}$.  Any assumption that $H_{A_{th} B}$ has no effect on the state components of $B$ from which the $A_i$ obtain observational outcomes is a fine-tuning assumption.

\subsection{Causal models induce contextuality} \label{causal-models}

When QRFs are considered to be implemented by CCCDs, and hence as effectively being networks of logical gates, it is clear that QRFs themselves -- indeed, any physically-implemented instruments, including human brains -- constitute (partial) causal models of the general environment $B$ (see \cite{frisch:14} for general discussion).  Coarse-graining of classically recorded outcomes (or ``amplification'' in the language of \cite{bohr:28}) can also induce {\em causal emergence}, a form of (informational) symmetry breaking describable in terms of {\em effective information} \cite{tononi:03} as developed in \cite{hoel:13,hoel:17}.  When formulated in terms of the Shannon noisy-channel coding theorem \cite{cover:06}, causal emergence occurs when a coarse-grained or macroscale description (a ``map'') of a system is more informative about its (classical) causal capacity than a microscale description.

Any finite causal model can be represented as a finite directed acyclic graph (DAG) via the causal Markov condition (CMC, \cite{geiger:90,pearl:09}).  Any such finite DAG is noncontextual; it specifies a task environment that is free, by definition, of unspecified side effects.  Hence as shown by \cite{cavalcanti:18}, imposing a finite causal model on any physical situation amounts to a fine-tuning assumption: all influences of the larger environment on the modeled environment are assumed to cancel out.  Here entanglement is ruled out by requiring {\em no disturbance} \cite[Def 2]{cavalcanti:18}, i.e. that $P(A \vert XY) = P(A \vert X)$ and $P(B \vert XY) = P(B \vert Y)$, for all values of the variables $A,B,X,Y$ for which these conditionals are defined.\footnote{This sense of `no disturbance' has also been called `no-signaling', where in relativistic accordance, such pairs of measurements $X,Y$ are made in spacelike separated regions (see also e.g. \cite{abramsky:14b,frembs:19}). }  This suggestion of {\em causally-induced contextuality} is affirmed in \cite{gogioso:23a,gogioso:23b}, where ambient causal constraints themselves become context dependent.  The approach taken in \cite{gogioso:23a,gogioso:23b} is analogous, in part, to that of \cite{abramsky:14b}, in adapting constraint functions and input histories to create a presheaf of causal distributions, but departs from \cite{abramsky:14b} in so far that under certain conditions, the constructed presheaf does not form a sheaf.

Within the framework of Deep Learning (DL), any trained network implements a causal model.  A precondition to prove algorithmic effectiveness is the causal faithfulness condition (CFC); any causal Bayesian network satisfying the CMC and CFC admits a unique Markov blanket (MB, \cite{guyon:03,guyon:07,pellet:08}). In a Bayesian network, the MB of any class variable provides all of the information required to predict its value.  This allows an optimal solution of the Feature Selection (FS) problem which, if solvable, offers advantages including data compression and effectiveness in prediction. From an algorithmic aspect, introducing causality into FS may enhance robustness and mechanistic activity, upon which causal analysis may determine whether or not the FS solution obtained is causally informative, or if some part of it is spurious \cite{guyon:03,guyon:07,pellet:08,guo:22}. More generally, there are several classical error correction programs analogous to quantum error correction, or possibly derivable from the latter in the classical limit, where effectively the presence of a holographic screen implements an MB (see discussions in \cite{fgm:22b,fgm:23,fields-et-al:23}); in particular, those correction/error-awareness algorithms suitable for FS that identify the MB of the class variable (e.g. \cite{guo:22, aliferis:03,koller:96}).

\section{Conclusion} \label{concl}

We have seen here how the phenomenon of contextuality in physical preparations or measurements follows from the separability -- non-entanglement -- of observers or their components, and how the undecidability of the QFP similarly hinges on separability.  Practical measurements of entanglement by separable observers -- e.g. in Bell/EPR experiments -- avoid these consequences via a fine-tuning assumption.  Such assumptions can invoke ``preferred'' or {\em a priori} selected QRFs, assumptions that thermodynamic exchange is decoupled from measurement in the observer, the environment, or both, or more general impositions of {\em a priori} causal models that circumscribe the task environment.  The results of \S \ref{QFP-undec}, in particular, show that no observer can determine that a causal model captures all the potential side-effects of any action.

Abramsky \cite{abramsky:18} has previously considered the relationship between contextuality and ``liar paradox'' sentences, e.g. ``this sentence is false.''  Such sentences play a key role in proofs of G\"{o}del's 1st incompleteness theorem, where such a sentence is proven to be true if and only if it is not provable within a particular axiomatic system \cite{godel:31}.  As a principled restriction on the applicability of classical causal models, which like proofs comprise sequences of sentences connected by law-like relations, contextuality speaks, like G\"{o}del's theorem, to the limitations of logical or lawful deduction.  Corollary \ref{isolation-cor} is particularly G\"{o}delean, as it shows that some side effects (true sentences) cannot be predicted by any given causal model (proof).  The insufficiency of ``laws'' as a basis for physics in the quantum era has of course been noted previously \cite{wheeler:83}.

While contextuality is often (e.g. in \cite{abramsky:18}) regarded as paradoxical, when the conditions for contextuality are restated in terms of the FP or the QFP, they appear somewhat inevitable.  A solution to either the FP or the QFP would, indeed, seem to require an infinite or omniscient observer \cite{dietrich:96, dietrich:20}.  The fine-tuning assumptions that crop up ubiquitously in situations described as noncontextual reinforce this, suggesting that it is the assumption of noncontextuality -- and the Laplacian omniscience it seems to require -- that is paradoxical. It is a paradox which carries with it an irreconcilable difference between operational and ontological models. Perhaps that is one reason why quantum theory -- when viewed as a {\em mechanical} theory of moving objects -- remains up to the present day a theory that falls short of being universally comprehendible, as R. Feynman has famously asserted.


\section*{Conflict of interest}
The authors report no conflicts of interest involved in this work.

\section*{Acknowledgement}
We thank Eric Dietrich and Antonino Marcian\`{o} for prior discussions on the frame problem and contextuality, respectively.

\section*{Funding}
This work received no external funding.


\end{document}